\documentclass[12pt]{article}

\usepackage[a4paper,left=20mm,top=20mm,right=20mm,bottom=25mm]{geometry}

\usepackage{graphicx}
\usepackage{algorithm}
\usepackage[noend]{algpseudocode}
\usepackage{amsmath, amsfonts, amssymb, amsthm}
\usepackage{mathtools}

\usepackage{authblk}
\usepackage[numbers, sort&compress]{natbib}

\usepackage[colorlinks]{hyperref}
\hypersetup{
	citecolor=blue,
   linkcolor=red,
}
\usepackage[font=footnotesize,width=.85\textwidth,labelfont=bf]{caption}
\usepackage[mathcal]{euscript}
\usepackage{mathptmx}

\usepackage{color}

\newcommand{\bul}{\mathfrak{s}}
\newcommand{\squ}{\mathfrak{d}}
\newcommand{\tri}{\mathfrak{t}}

\newcommand{\mc}{\mathcal}
\newcommand{\wh}{\widehat}
\newcommand{\Rd}{\ensuremath{\mathcal{R}_{\sigma}}}
\newcommand{\LT}{\ensuremath{\mathcal{L}}}
\newcommand{\lca}{\ensuremath{\operatorname{lca}}}
\newcommand{\Gen}{\ensuremath{\mathbb{G}}}
\newcommand{\Spe}{\ensuremath{\mathbb{S}}}
\newcommand{\rt}[1]{\ensuremath{\mathsf{(#1)}}}
\newcommand{\EHGT}{\ensuremath{\mathcal{E}}}
\newcommand{\Th}{\ensuremath{T_{\mathcal{\overline{E}}}}}
\newcommand{\sT}{\sigma_{\Th}}
\newcommand{\tT}{\ensuremath{\tau_T}}

\providecommand{\keywords}[1]{\textbf{\textit{Keywords: }} #1}

\newtheorem{theorem}{Theorem}[section]
\newtheorem{lemma}{Lemma}[section]

\newtheorem{remark}{Remark}

\newtheorem{definition}{Definition}

\makeindex             

\begin{document}

\title{Biologically Feasible Gene Trees, Reconciliation Maps and Informative Triples}

\author[]{Marc Hellmuth}

\affil[]{\footnotesize Dpt.\ of Mathematics and Computer Science, University of Greifswald, Walther-
  Rathenau-Strasse 47, D-17487 Greifswald, Germany\\ \smallskip

	Saarland University, Center for Bioinformatics, Building E 2.1, P.O.\ Box 151150, D-66041 Saarbr{\"u}cken, Germany\\ \smallskip
	Email: \texttt{mhellmuth@mailbox.org}}

\date{\ }

\maketitle

\abstract{ 
The history of gene families - which are equivalent to \emph{event-labeled} gene trees - can be reconstructed from empirically  estimated evolutionary event-relations  containing pairs of orthologous, paralogous or xenologous genes.  The question then arises as whether inferred event-labeled gene trees are \emph{biologically feasible}, that is, if there is a possible true history that would explain a given gene tree.  In practice, this problem is boiled down to finding a reconciliation map - also known as DTL-scenario - between the event-labeled gene trees and a (possibly unknown) species tree.

In this contribution, we first characterize whether there is a valid
reconciliation map for binary event-labeled gene trees $T$ that contain
speciation, duplication and horizontal gene transfer events and some unknown
species tree $S$ in terms of ``informative'' triples that are displayed in $T$
and provide information of the topology of $S$. These informative triples are used to 
infer the unknown species tree $S$ for $T$.
We obtain a similar result for non-binary gene trees. To this end, however, 
the reconciliation map needs to be further restricted. 
We provide a
polynomial-time algorithm to decide whether there is a species tree for a given
event-labeled gene tree, and in the positive case, to construct the
species tree and the respective (restricted) reconciliation map.

However, informative triples as well as DTL-scenarios have its limitations
when they are used to explain the biological feasibility of gene trees. 
While reconciliation maps imply biological feasibility, 
we show that the converse is not true in general.  Moreover, we show that 
informative triples do neither provide
enough information to characterize ``relaxed'' DTL-scenarios nor 
non-restricted reconciliation maps for 
non-binary biologically feasible gene trees.

}

\smallskip
\noindent

\sloppy

\section{Background}
\label{sec:intro}

The evolutionary history of genes is intimately linked with the history of the
species in which they reside. Genes are passed from generation to generation to
the offspring. Some of those genes are frequently duplicated, mutate, or get
lost - a mechanism that also ensures that new species can evolve. In particular,
genes that share a common origin (\emph{homologs}) can be classified into the
type of their ``evolutionary event relationship'', namely \emph{orthologs}, \emph{paralogs}
and \emph{xenologs} \cite{Gray:83,Fitch2000}. Two homologous genes are
\emph{orthologous} if at their most recent point of origin the ancestral gene is
transmitted to two daughter lineages; a \emph{speciation} event happened. They
are \emph{paralogous} if the ancestor gene at their most recent point of origin
was duplicated within a single ancestral genome; a \emph{duplication} event
happened. Horizontal gene transfer (HGT) refers to the transfer of genes between
organisms in a manner other than traditional reproduction and across different
species and yield so-called \emph{xenologs}.
In contrast to orthology and paralogy, the definition of xenology is less
well established and by no means consistent in the biological
literature. One definition stipulates that two genes are
\emph{xenologs} if their history since their common ancestor involves
horizontal transfer of at least one of them \cite{Fitch2000,Jensen:01}. 
The mathematical framework for evolutionary event-relations 
relations in terms of symbolic
ultrametrics, cographs and 2-structures \cite{Boeckner:98,HHH+13,HSW:16,HW:16b}, on the other hand, naturally
accommodates more than two types of events associated with the internal nodes of
the gene tree. We follow the notion in 
\cite{Gray:83, HSW:16} and call two genes xenologous,  whenever 
their least common ancestor was a HGT event.  

The knowledge of evolutionary event relations such as orthology, paralogy or
xenology is of fundamental importance in many fields of mathematical and
computational biology, including the reconstruction of evolutionary
relationships across species \cite{DBH-survey05,GCMRM:79,HHH+12,Hellmuth:15a,Lafond2014},
as well as functional genomics and gene organization in species
\cite{GK13,TG+00,TKL:97}. Intriguingly, there are methods to infer orthologs
\cite{Altenhoff:09,AGGD:13,ASGD:11,CMSR:06,Lechner:11a,Lechner:14,inparanoid:10,TG+00,T+11,Wheeler:08}
or to detect HGT \cite{CBRC:02,Dessimoz2008,LH:92,PMT+99,RSLD15} \emph{without}
the need to construct gene or species trees. Given  empirical estimated
event-relations one can infer the history of gene families which are
equivalent to event-labeled gene trees \cite{HHH+13,HSW:16,Hellmuth:15a, 
lafond2015orthology,DEML:16,LDEM:16}. For an
event-labeled gene tree to be biologically feasible there must be a putative
``true'' history that can explain the observed gene tree. However, in practice
it is not possible to observe the entire evolutionary history as e.g.\ gene
losses eradicate the entire information on parts of the history. Therefore, in
practice the problem of determining whether an event-labeled gene tree is
biologically feasible is reduced to the problem of finding a valid
reconciliation map, also known as DTL-scenario, between the event-labeled gene
trees and an arbitrary (possibly unknown) species tree. Tree-reconciliation
methods have been extensively studied over the last years
\cite{GJ:06,RLG+14,DCH:09,DES:14,VSGD:08,lafond2012optimal,SLX+12,huson2011survey,page1998genetree,
STDB:15,SD:12,MMZ:00,DRDB11,Doyon2010,EHL10,GCMRM:79,THL:11}
and are often employed to identify inner vertices of the gene tree as a
duplication, speciation or HGT, given that both, the gene and the species tree
are available.

In this contribution, we assume that only the event-labeled gene tree $T$ is available and wish
to answer the question: How much information about the species tree $S$ and the 
reconciliation  between $T$ and $S$ is already contained in the gene tree $T$?
As we shall see, this question can easily be answered for \emph{binary} gene trees in terms
of ``informative'' triples that are displayed in $T$ and provide
information on the topology of $S$. The latter generalizes results established
by Hernandez et al.\ \cite{HHH+12} for the HGT-free case. 
To obtain a similar result for \emph{non-binary} gene trees, we show
that the reconciliation map needs to be restricted. 
Nevertheless, informative triples can then  be used to characterize whether
there is a valid restricted reconciliation map for a given non-binary
gene tree and some unknown species tree $S$, as well as to construct $S$,
provided the informative triples are consistent. 
However, this approach
has also some limitations. We prove that 
``informative'' triples are not sufficient to characterize the existence
of  a possibly ``relaxed'' reconciliation map.
Moreover, while reconciliation maps give clear evidence of gene trees to be biologically feasible, 
the converse is in general not true. We provide a simple example that shows that not all
biologically feasible gene trees can be explained by DTL-scenarios. 

\section{Preliminaries}
\label{sec:prelim}


A \emph{rooted tree $T=(V,E)$ (on $L$)} is an acyclic connected simple graph
with leaf set $L\subseteq V$, set of edges $E$, and set of interior vertices
$V^0=V\setminus L$ such that there is one distinguished vertex $\rho_T \in V$,
called the \emph{root of $T$}. 

A vertex $v\in V$ is called a \emph{descendant} of $u\in V$, $v \preceq_T u$, and $u$ is an
\emph{ancestor} of $v$, $u \succeq_T v$, if $u$ lies on the path from
$\rho_T$ to $v$. As usual, we write $v \prec_T u$ and $u \succ_T v$ to
mean $v \preceq_T u$ and $u\ne v$.  
If $u \preceq_T v$ or $v \preceq_T u$ then $u$ and $v$
are \emph{comparable} and otherwise, \emph{incomparable}.  
For $x\in V$, we write $L_T(x):=\{ y\in L| y\preceq x\}$ for the
set of leaves in the subtree $T(x)$ of $T$ rooted in $x$.

\begin{remark}
It will be convenient to use a notation for edges $e$ that implies which of the vertex in
$e$ is closer to the root. Thus, the notation for edges $(u,v)$ of a tree
is always chosen such that $u\succ_T v$. 
\end{remark}

For our discussion below we need to extend the ancestor relation $\preceq_T$ on
$V$ to the union of the edge and vertex sets of $T$. More precisely, for the
edge $e=(u,v)\in E$ we put $x \prec_T e$ if and only if $x\preceq_T v$ and $e
\prec_T x$ if and only if $u\preceq_T x$. For edges $e=(u,v)$ and $f=(a,b)$ in
$T$ we put $e\preceq_T f$ if and only if $v \preceq_T b$. In the latter case, 
the edges $e$ and $f$ are called comparable.

For a non-empty subset of leaves $A\subseteq L$, we define $\lca_T(A)$, or the
\emph{least common ancestor of $A$}, to be the unique $\preceq_T$-minimal vertex
of $T$ that is an ancestor of every vertex in $A$. In case $A=\{x,y \}$, we put
$\lca_T(x,y):=\lca_T(\{x,y\})$ and if $A=\{x,y,z \}$, we put
$\lca_T(x,y,z):=\lca_T(\{x,y,z\})$. For later reference, note that for all $x\in
V$ it hold that $x=\lca_T(L_T(x))$. We will also make frequent use that for two
non-empty vertex sets $A,B$ of a tree, it always holds that $\lca(A\cup B) =
\lca(\lca(A),\lca(B))$. 

A \emph{phylogenetic tree $T$ (on $L$)} is a rooted tree $T=(V,E)$ (on $L$) such
that no interior vertex  $v\in V^0$ has degree two, except possibly the root 
$\rho_T$. If $L$ corresponds to a \emph{set of genes} $\Gen$
or \emph{species} $\Spe$, we call a phylogenetic tree on $L$ \emph{gene tree}
and \emph{species tree}, respectively. The \emph{restriction} $T|_{L'}$ of of a
phylogenetic tree $T$ to $L'\subseteq L$ is the rooted tree with leaf set $L'$
obtained from $T$ by first forming the minimal spanning tree in $T$ with leaf
set $L'$ and then by suppressing all vertices of degree two with the exception
of $\rho_T$ if $\rho_T$ is a vertex of that tree.


Rooted triples are phylogenetic trees on three leaves with precisely two
interior vertices. They constitute an important concept in the context of
supertree reconstruction \cite{Semple:book,Bininda:book,Dress:book} and will
also play a major role here. A rooted tree $T$ on $L$ \emph{displays} a triple
$\rt{xy|z}$ if, $x,y,z\in L$ and the path from $x$ to $y$ does not intersect the
path from $z$ to the root $\rho_T$ and thus, having $\lca_T(x,y)\prec_T
\lca_T(x,y,z)$. We denote by $\mc{R}(T)$ the set of all triples that are
displayed by the rooted tree $T$.  

A set $R$ of triples is \emph{consistent} if there is a rooted tree $T$ on $L_R=
\cup_{r\in R} L_r(\rho_r)$ such that $R\subseteq \mc{R}(T)$ and thus, $T$
\emph{displays} each triple in $R$. Not all sets of triples are consistent of
course. Nevertheless, given a triple set $R$ there is a polynomial-time
algorithm, referred to in \cite{Semple:book,Steel:book} as \texttt{BUILD}, that
either constructs a phylogenetic tree $T$ that displays $R$ or that recognizes
that $R$ is not consistent \cite{Aho:81}. 
The runtime of BUILD is $O(|L_R||R|)$ \cite{Semple:book}. 
Further practical implementations and improvements have been
 discussed in  \cite{Jansson:05,DF:16,Henzinger:99,Holm:01}.

We will consider rooted trees
$T=(V,E)$ from which particular edges are removed. Let $\EHGT\subseteq E$ and
consider the forest $\Th\coloneqq (V,E\setminus \EHGT)$. We can preserve the
order $\preceq_T$ for all vertices within one connected component of $\Th$ and
define $\preceq_{\Th}$ as follows: $x\preceq_{\Th}y$ iff $x\preceq_{T}y$ and
$x,y$ are in same connected component of $\Th$. Since each connected component
$T'$ of $\Th$ is a tree, the ordering $\preceq_{\Th}$ also implies a root
$\rho_{T'}$ for each $T'$, that is, $x\preceq_{\Th} \rho_{T'}$ for all $x\in
V(T')$. If $L(\Th)$ is the leaf set of $\Th$, we define $L_{\Th}(x) = \{y\in
L(\Th) \mid y\prec_{\Th} x\}$ as the set of leaves in $\Th$ that are reachable
from $x$. Hence, all $y\in L_{\Th}(x)$ must be contained in the same connected
component of $\Th$. We say that the forest $\Th$ displays a triple $r$, if $r$
is displayed by one of its connected components. Moreover, $\mc{R}(\Th)$ denotes
the set of all triples that are displayed by the forest $\Th$.

\section{Biologically Feasible and Observable Gene Trees}

A gene tree  arises through a series of events (speciation, duplication, HGT,
and gene loss) along a species tree. In a ``true history'' the gene tree
$\wh{T} = (V,E)$ on a set of genes $\wh\Gen$ is equipped with an 
\emph{event-labeling} map $t:V\cup E\to \wh I\cup \{0,1\}$ with
$\wh I=\{\bul,\squ,\tri,\odot,\textrm{x}\}$ that assigns to each
vertex $v$ of $\wh T$ a value $t(v)\in \wh I$ indicating whether $v$ is 
a speciation event ($\bul$), duplication event ($\squ$), HGT event
($\tri$), extant leaf ($\odot$) or a loss event ($\textrm{x}$).
In addition, to each edge $e$ a value $t(e)\in \{0,1\}$ is added that
indicates whether $e$ is a \emph{transfer edge} ($1$) or not ($0$).
Note, in the figures we used the symbols $\bullet, \square$ and $\triangle$
for $\bul, \squ$ and $\tri$ respectively.
Hence, $e=(x,y)$ and $t(e) =1$ iff $t(x)=\tri$ and the genetic material is
transferred from the species containing $x$ to the species containing $y$. 
We remark that the restriction $t_{|V}$ of $t$ to the vertex set $V$ was introduced
as ``symbolic dating map'' in \cite{Boeckner:98} and that there is a close
relationship to so-called cographs \cite{HHH+13,HW:15,HW:16a}.
Let $\Gen \subseteq \wh{\Gen}$ be the set of all extant genes in $\wh T$. 
Hence,  there is a map $\sigma:\Gen\to
\Spe$ that assigns to each extant gene the extant species in which it resides.

We assume that the gene tree and its event labels are inferred from
(sequence) data, i.e., $T$ is restricted to those labeled trees that can be
constructed at least in principle from observable data. 
Gene losses eradicate the entire information on parts of the history and thus, 
cannot directly be observed from extant sequences.  
Hence, in our setting the (observable) gene tree $T$ is the restriction  $\wh{T}_{|\Gen}$ to
the set of extant genes, see Figure \ref{fig:true}.
Since all leaves of $T$ are extant genes in $\Gen$ we don't need to specially
label the leaves in \Gen, and thus simplify the 
event-labeling map $t:V^0\cup E\to I\cup \{0,1\}$ by assigning only to the 
interior vertex an event in $I=\{\bul,\squ,\tri\}$. 
We assume here that all non-transfer edges transmit the genetic material vertically, 
that is, from an ancestral species to its descendants.

\begin{definition}
We write $(T;t,\sigma)$ for the  tree $T=(V,E)$ with event-labeling $t$
and corresponding map $\sigma$. The set $\EHGT = \{e\in E\mid t(e)=1\}$ 
will always denote  the set of transfer edges in  $(T;t,\sigma)$. 

Additionally, we consider gene tree $(T=(V,E);t,\sigma)$ from which the transfer edges have been removed, resulting
in the forest $\Th = (V, E\setminus \EHGT)$ in which we preserve the event-labeling $t$,
that is, we use the restriction $t_{|V}$ on $\Th$. 
\end{definition}

We call a gene tree $(T;t,\sigma)$ on $\Gen$ \emph{biologically feasible}, 
if there is a true scenario such that $T = \wh{T}_{|\Gen}$, that is, 
there is a true history that can explain  $(T;t,\sigma)$. 
By way of example, the gene tree in Figure \ref{fig:true}(right) 
is biologically feasibly. 
However, so-far it is unknown whether there are gene trees $(T;t,\sigma)$
that are not biologically feasible. 
Answering the latter might be a hard task, as many HGT or duplication vertices 
followed by losses can be inserted into $T$ that may result in a putative
true history that explains the event-labeled gene tree.

\begin{figure*}[tbp]
  \begin{center}
    \includegraphics[width=.7\textwidth]{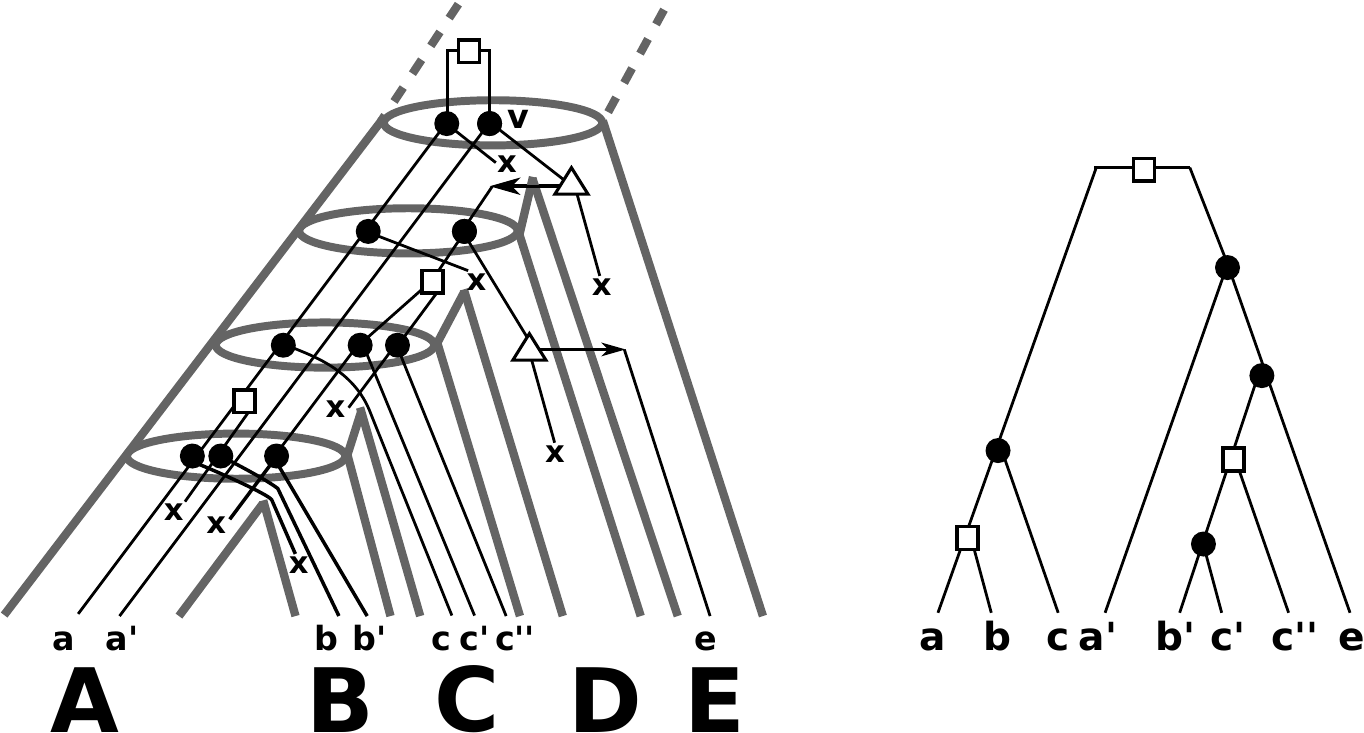}
  \end{center}
	\caption{\emph{Left:} 
		An example of a ``true'' history of a gene tree that evolves along the
		(tube-like) species tree. The set of
		extant genes $\Gen$ comprises $a,a',b,b',c,c',c''$ and $e$ and $\sigma$ maps each
		gene in $\Gen$ to the species (capitals below the genes) $A,B,C,E\in
		\sigma(\Gen)$. For simplicity all speciation events followed by a loss along
		the path from $v$ to $a'$ in $T$ are omitted.
	   \emph{Left:} The observable gene tree $(T;t,\sigma)$ is shown. Since there 
		is a true scenario which explains $(T;t,\sigma)$, the gene tree is biologically feasible.
		In particular, $(T;t,\sigma)$ satisfies (O1), (O2) and (O3).
    }
	\label{fig:true}
\end{figure*}

Following N{\o}jgaard et al.\ \cite{N+17}, we additionally restrict the set of observable gene trees $(T;t,\sigma)$ 
to those gene trees that satisfy the following observability axioms:\smallskip

\begin{description}
\item[(O1)] Every internal vertex $v$ has degree at least $3$, except
  possibly the root which has degree at least $2$. 
\item[(O2)] Every HGT node has at least one transfer edge, $t(e)=1$, and at
  least one non-transfer edge, $t(e)=0$; 
\item[(O3)] \emph{\textbf{(a)}} 
If $x\in V$ is a speciation vertex, 
then there are distinct children $v$,
$w$ of $x$ in $T$ with 
	$\sT(v)\cap \sT(w) = \emptyset$. \\
\emph{\textbf{(b)}} 
If $(x,y) \in \EHGT$, then 
	$\sT(x)\cap \sT(y) = \emptyset$.
\end{description} \smallskip
Condition (O1) is justified by the restriction $T=\wh{T}_{|\Gen}$ 
of the true binary gene tree $\wh{T}$ to the set of extant genes $\Gen$, 
since  $T=\wh{T}_{|\Gen}$ is always a \emph{phylogenetic} tree. 
In particular, (O1) ensures that every event leaves a historical trace in the
sense that there are at least two children that have survived in at least
two of its subtrees. 
Condition (O2) ensures that for
an HGT event a historical trace remains of both the transferred and the
non-transferred copy.  

Condition (O3.a) is a consequence of (O1), (O2) and a stronger condition (O3.a') claimed in \cite{N+17}:
\emph{If $x$ is a speciation vertex, then there
are at least two distinct children $v,w$ of $x$ such that the species $V$
and $W$ that contain $v$ and $w$, resp., are incomparable in $S$.}
Note, a speciation vertex $x$ cannot be observed from data
if it does not ``separate'' lineages, that is, there are two leaf
descendants of distinct children of $x$ that are in distinct
species. Condition (O3.a') is even weaker and 
ensures that any ``observable'' speciation vertex $x$ separates at
least locally two lineages. As a result of  (O3.a') one can obtain
(O3.a) \cite{N+17}.
Intuitively, (O3.a) is satisfied since within a connected component of $\Th$ no
genetic material is exchanged between non-comparable nodes. Thus, a gene
separated in a speciation event necessarily ends up in distinct species in
the absence of the transfer edges.

Condition (O3.b) is a consequence of (O1), (O2) and a stronger condition (O3.b') claimed in \cite{N+17}:
\emph{If $(v,w)$ is a transfer edge in $T$, then   $t(v)=\tri$ and the
species $V$ and $W$ that contain $v$ and $w$, resp., are
incomparable in $S$.}
Note, if $(v,w)\in \EHGT$ then $v$ signifies the transfer event
itself but $w$ refers to the next (visible) event in the gene tree
$T$.  In a ``true
history'' $v$ is contained in a species $V$ that transmits its genetic
material (maybe along a path of transfers) to a contemporary species $Z$
that is an ancestor of the species $W$ containing $w$. In order to have
evidence that this transfer happened, Condition (O3.b') is used and
as a result one obtains (O3.b). 
The intuition behind  (O3.b) is as follows:
Observe that  ${\Th}(x)$ and ${\Th}(y)$ are subtrees of distinct connected
components of $\Th$ whenever $(x,y) \in \EHGT$. 
Since HGT amounts to the transfer of genetic material
\emph{across} distinct species, the genes $x$ and $y$ are in distinct
species, cf.\ (O3.b). However, since $\Th$ does not contain transfer edges  and thus, there is no
genetic material transferred across distinct species \emph{between} distinct
connected components in $\Th$. We refer to \cite{N+17} for further details

\begin{remark}
In what follows, we only consider  gene trees $(T;t,\sigma)$
that satisfy (O1), (O2) and (O3).

We simplify the notation a bit and write $\sT(u):=\sigma(L_{\Th}(u))$.
\end{remark}

Based on Axiom (O2) the following results was established in \cite{N+17}.
\begin{lemma}
Let $(T;t,\sigma)$ be an event-labeled gene tree.
Let $\mathcal{T}_1, \dots, \mathcal{T}_k$ be
the connected components of $\Th$ with roots $\rho_1, \dots, \rho_k$,
respectively. 
Then, $\{L_{\Th}(\rho_1), \dots, L_{\Th}(\rho_k)\}$ forms a partition of $\Gen$. 
\label{lem:part_gen}
\end{lemma}

Lemma  \ref{lem:part_gen} particularly implies that 
$\sT(x) \neq \emptyset$ for all $x\in V(T)$. Note, 
$\Th$ might contain interior vertices (distinct from the root)
that have degree two. Nevertheless, for each $x\preceq_{\Th} y$ in
$\Th$ we have $x\preceq_T y$ in $T$. Hence, partial information (that in
particular is ``undisturbed'' by transfer edges) on the partial ordering of the
vertices in $T$ can be inferred from $\Th$.

\section{Reconciliation Map}

Before we define a reconciliation map that ``embeds'' a
given gene tree into a given species tree we need a slight modification 
of the species tree.
In order to account for duplication events
that occurred before the first speciation event, we need to add an extra vertex
and an extra edge ``above'' the last common ancestor of all species: hence, we
add an additional vertex to $W$ (that is know the new root $\rho_S$ of $S$) and the
additional edge $(\rho_S,\lca_S(\Spe))\in F$. Note that strictly speaking $S$ is
not a phylogenetic tree anymore. In case there is no danger of confusion, we
will from now on refer to a phylogenetic tree on $\Spe$ with this extra edge and
vertex added as a species tree on $\Spe$.

\begin{definition}[DTL-scenario] \label{def:mu} 
  Suppose that $\Spe$ is a set of species, $S=(W,F)$ is a phylogenetic tree on
  $\Spe$, $T=(V,E)$ is a gene tree with leaf set $\Gen$ and that $\sigma:\Gen\to
  \Spe$ and $t:V^0\to \{\bul,\squ,\tri\} \cup \{0,1\}$ are the
  maps described above. Then we say that \emph{$S$ is a species tree for
  $(T;t,\sigma)$} if there is a map $\mu:V\to W\cup F$ such that, for all $x\in
  V$:  
\begin{description}
\item[(M1)] \emph{Leaf Constraint.}  If $x\in \Gen$   then $\mu(x)=\sigma(x)$. 		\vspace{0.03in}
\item[(M2)] \emph{Event Constraint.}
	\begin{itemize}
		\item[(i)]  If $t(x)=\bul$, then  
                  $\mu(x) = \lca_S(\sT(x))$.
		\item[(ii)] If $t(x) \in \{\squ, \tri\}$, then $\mu(x)\in F$. 
		\item[(iii)] If $t(x)=\tri$ and $(x,y)\in \EHGT$, 
						 then $\mu(x)$ and $\mu(y)$ are incomparable in $S$. 
 	\end{itemize} \vspace{0.03in}
\item[(M3)] \emph{Ancestor Constraint.}		\\	
		Let $x,y\in V$ with $x\prec_{\Th} y$. 
		Note, the latter implies that the path connecting $x$ and $y$ in $T$
		does not contain transfer edges.
		We distinguish two cases:
	\begin{itemize}
		\item[(i)] If $t(x),t(y)\in \{\squ, \tri\}$, then $\mu(x)\preceq_S \mu(y)$, 
		\item[(ii)] otherwise, i.e., at least one of $t(x)$ and $t(y)$ is a speciation $\bul$, 
						$\mu(x)\prec_S\mu(y)$.
	\end{itemize}
\end{description}
We call $\mu$ the \emph{reconciliation map} from $(T;t,\sigma)$ to $S$. 
\end{definition}                

Definition \ref{def:mu} is a natural generalization of the map defined in
\cite{HHH+12}, that is, in the absence of horizontal gene transfer, Condition
(M2.iii) vanishes and thus, the proposed reconciliation map precisely coincides
with the one given in \cite{HHH+12}.
In case that the event-labeling of $T$ is unknown, but a species tree $S$ is
given, the authors in \cite{THL:11,BAK:12} gave an axiom set, called
DTL-scenario, to reconcile $T$ with $S$. This reconciliation is then used to
infer the event-labeling $t$ of $T$. 
The ``usual'' DTL axioms explicitly refer to binary,
fully resolved gene and species trees. We therefore use a different axiom set
that is, nevertheless, equivalent to DTL-scenarios
in case the considered gene trees are binary \cite{N+17}.

Condition (M1) ensures that each leaf of $T$, i.e., an extant gene in $\Gen$, is
mapped to the species in which it resides. 
Condition (M2.i) and (M2.ii) ensure
that each vertex of $T$ is either mapped to a vertex or an edge in $S$ such that
a vertex of $T$ is mapped to an interior vertex of $S$ if and only if it is a
speciation vertex. We will discuss (M2.i) in further detail below. 
Condition (M2.iii) maps the vertices of a transfer edge
in a way that they are incomparable in the species tree and is used
to satisfy axiom (O3). 
Condition (M3) refers only to the connected components of $\Th$ and
is used to preserve the ancestor order $\preceq_T$ of $T$ along the paths
that do not contain transfer edges is preserved. 

It needs to be discussed, why one \emph{should} map a speciation vertex $x$ to 
$\lca_S(\sT(x))$ as required in (M2.i). The next lemma shows, that one
\emph{can} put $\mu(x) = \lca_S(\sT(x))$.  
\begin{lemma}[\cite{N+17}]
	Let $\mu$ be a reconciliation map  from $(T;t,\sigma)$ to $S$
	that satisfies (M1) and (M3), then
   $\mu(u)\succeq_S \lca_S(\sT(u))$ for any $u\in V(T)$. 
\label{lem:cond-mu-lca}
\end{lemma}
Condition (M2.i) implies in particular the weaker property ``(M2.i') if
$t(x)=\bul$ then $\mu(x)\in W$''. In the light of
Lemma~\ref{lem:cond-mu-lca}, $\mu(x)=\lca_S(\sT(x))$ is the lowest possible
choice for the image of a speciation vertex. 
Instead of considering the
possibly exponentially many reconciliation maps for which
$\mu(x)\succ_S\lca_S(\sT(x))$ for speciation vertices $x$ is allowed
we restrict our attention to those that satisfy (M2.i) only. 
In particular, as we shall see later, there is a neat characterization
of maps that satisfy (M2.i) that does, however, 
not work for maps with ``relaxed'' (M2.i). 

Moreover, we have the following result, which is a mild
generalization of \cite{N+17}. 
\begin{lemma}
	Let $\mu$ be a reconciliation map from a gene tree $(T;t,\sigma)$ to $S$. 
	\begin{enumerate}
		\item 	If $v,w\in V(T)$ are in the same connected component of $\Th$,   then 
					$\mu(\lca_{\Th}(v,w)) \succeq_S \lca_S(\mu(v),\mu(w))$.
					\label{item:lca}
		\item 	If $(T;t,\sigma)$ is a binary gene tree and $x$ a speciation vertex with children $v,w$ in $T$, then  
				   then $\mu(v)$ and $\mu(w)$ are incomparable in $S$. \label{item:speci} 
		\end{enumerate}
\label{lem:cond-mu}
\end{lemma}
\begin{proof}
	Let $v,w\in V(T)$ be in the same connected component of $\Th$. 
	Assume that $v$ and $w$ are comparable in $\Th$ and 
	that w.l.o.g.\ $v\succ_{\Th} w$. Condition (M3) implies that 
	$\mu(v)\succeq_S\mu(w)$. Hence, 
	$v = \lca_{\Th}(v,w)$ and 
	$\mu(v) = \lca_S(\mu(v),\mu(w))$ and we are done. 

	Now assume that  $v$ and $w$ are incomparable in $\Th$. 
	Consider the unique path $P$ connecting $w$ with $v$ in $\Th$. This path $P$
	is uniquely subdivided into a path $P'$ and a path $P''$ from
	$\lca_{\Th}(v,w)$ to $v$ and $w$, respectively. Condition (M3) implies that
	the images of the vertices of $P'$ and $P''$ under $\mu$, resp., are ordered
	in $S$ with regards to $\preceq_S$ and hence, are contained in the intervals
	$Q'$ and $Q''$ that connect $\mu(\lca_{\Th}(v,w))$ with $\mu(v)$ and
	$\mu(w)$, respectively. In particular, $\mu(\lca_{\Th}(v,w))$ is the largest
	element (w.r.t.\ $\preceq_S$) in the union of $Q'\cup Q''$ which contains the
	unique path from $\mu(v)$ to $\mu(w)$ and hence also $\lca_S(\mu(v),\mu(w))$.
		
	Item \ref{item:speci} was already proven in \cite{N+17}.
\end{proof}

Assume now that there is a reconciliation map $\mu$ from $(T;t,\sigma)$ to $S$. 
From a biological point of view, however, it is necessary to reconcile a gene
tree with a species tree such that genes do not ``travel through time'', a
particular issue that must be considered whenever $(T;t,\sigma)$  contains HGT, 
see Figure \ref{fig:least} for an example. 

\begin{definition}[Time Map]\label{def:time-map}
  The map $\tT: V(T) \rightarrow \mathbb{R}$ is a time map for the 
  rooted tree $T$ if $x\prec_T y$ implies $\tT(x)>\tT(y)$ for all 
  $x,y\in V(T)$. 
\end{definition}

\begin{definition} \label{def:tc-mu} A reconciliation map $\mu$ from
  $(T;t,\sigma)$ to $S$ is \emph{time-consistent} if there are time maps
  $\tau_T$ for $T$ and $\tau_S$ for $S$ for all $u\in V(T)$ satisfying the
  following conditions:
  \begin{description}
  \item[(T1)] If $t(u) \in \{\bul, \odot \}$, then 
    $\tau_T(u) = \tau_S(\mu(u))$. \label{bio1}
  \item[(T2)] If $t(u)\in \{\squ,\tri \}$ and, thus
    $\mu(u)=(x,y)\in E(S)$, \label{bio2} then
    $\tau_S(y)>\tau_T(u)>\tau_S(x)$. 
\end{description}
\end{definition} 

Condition (T1) is used to identify the time-points of speciation vertices 
and leaves $u$ in the gene tree with the time-points of their respective 
images $\mu(u)$ in the species trees. Moreover, 
duplication or HGT vertices $u$ are mapped to edges $\mu(u)=(x,y)$ in $S$
and the time point of $u$ must thus lie between the time
points of $x$ and $y$ which is ensured by Condition (T2).
N{\o}jgaard et al.\ \cite{N+17} designed an $O(|V(T)|\log(|V(S)|))$-time algorithm 
to check whether a given reconciliation map $\mu$ is time-consistent, and an algorithm with the same
time complexity for the construction of a time-consistent reconciliation
map, provided one exists.
Clearly, a necessary condition for the existence of time-consistent reconciliation maps from 
$(T;t,\sigma)$ to $S$ is the  existence of \emph{some} reconciliation map 
$(T;t,\sigma)$ to $S$. 
In the next section,  we first characterize the existence of reconciliation maps
and discuss open  time-consistency problems.


\section{From Gene Trees to Species Trees}

Since a gene tree $T$ is uniquely determined by its induced triple set
$\mathcal{R}(T)$, it is reasonable to expect that a lot of information on the
species tree(s) for $(T;t, \sigma)$ is contained in the images of the triples in
$\mathcal{R}(T)$, (or more precisely their leaves) under $\sigma$. However, not
all triples in $\mathcal{R}(T)$ are informative, see Figure \ref{fig:reconc} for
an illustrative example. In the absence of HGT, it has already been shown by
Hernandez-Rosales et al.\ \cite{HHH+12} that the informative triples $r\in
\mathcal{R}(T)$ are precisely those that are rooted at a speciation event and
where the genes in $r$ reside in three distinct species. However, in the
presence of HGT we need to further subdivide the informative triples as follows.

\begin{definition} \label{def:triples} 
	Let $(T;t,\sigma)$ be a given event-labeled gene tree with respective
	set of transfer-edges $\EHGT = \{e_1,\dots,e_h\}$  and $\Th$	as defined above. 
	We define 
	\begin{align*}
	\Rd(\Th) = \{\rt{ab|c} \in \mc{R}(\Th) \colon&\sigma(a),\sigma(b),\sigma(c) \\ &\textrm{are pairwise distinct} \}
	\end{align*}
	as the subset of all triples displayed in $\Th$ such that the leaves are from pairwise distinct species.	

	\noindent
	Let
	\begin{equation*}
		\mc{R}_0(\Th) \coloneqq \{\rt{ab|c} \in \Rd(\Th) \colon  t(\lca_{\Th}(a,b,c)) = \bul \}	
	\end{equation*}
	be the set of triples in  $\Rd(\Th)$ that are rooted at a speciation event.

	\noindent
	For each $e_i=(x,y) \in \EHGT$ define 
	\begin{equation*}	
		\begin{split}
			\mc{R}_i(\Th)	\coloneqq  \{ \rt{ab|c} \colon& \sigma(a),\sigma(b),\sigma(c)  \textrm{ are pairwise distinct} \\
								&\textrm{and either }  a,b\in L_{\Th}(x), c\in L_{\Th}(y) \\ &\text{or } c\in L_{\Th}(x), a,b\in L_{\Th}(y)  \}.
		\end{split}
	\end{equation*}	
	Hence, $\mc{R}_i(\Th)$ contains a triple $\rt{ab|c}$ for every 
	$a,b\in L_{\Th}(x), c\in L_{\Th}(y)$ that reside in pairwise distinct species. Analogously,  for any $a,b\in L_{\Th}(y), c\in L_{\Th}(x)$ there is a triple $\rt{ab|c}\in \mc{R}_i(\Th)$, if
	$\sigma(a),\sigma(b),\sigma(c)$ are pairwise distinct. \smallskip

	\noindent
	The \emph{informative triples} of $T$ are comprised in the set 	
	$\mathcal{R}(T;t,\sigma) = \cup_{i=0}^h \mc{R}_i(\Th)$.

	\noindent
	Finally, we define the informative species triple set 
	\[\mathcal{S}(T;t,\sigma)\coloneqq \{\rt{\sigma(a)\sigma(b)|\sigma(c)} \colon \rt{ab|c} \in \mathcal{R}(T;t,\sigma) \}\]
	that can be inferred from the informative triples of $(T;t,\sigma)$.
\end{definition}

\begin{figure*}[tbp]
  \begin{center}
    \includegraphics[width=.7\textwidth]{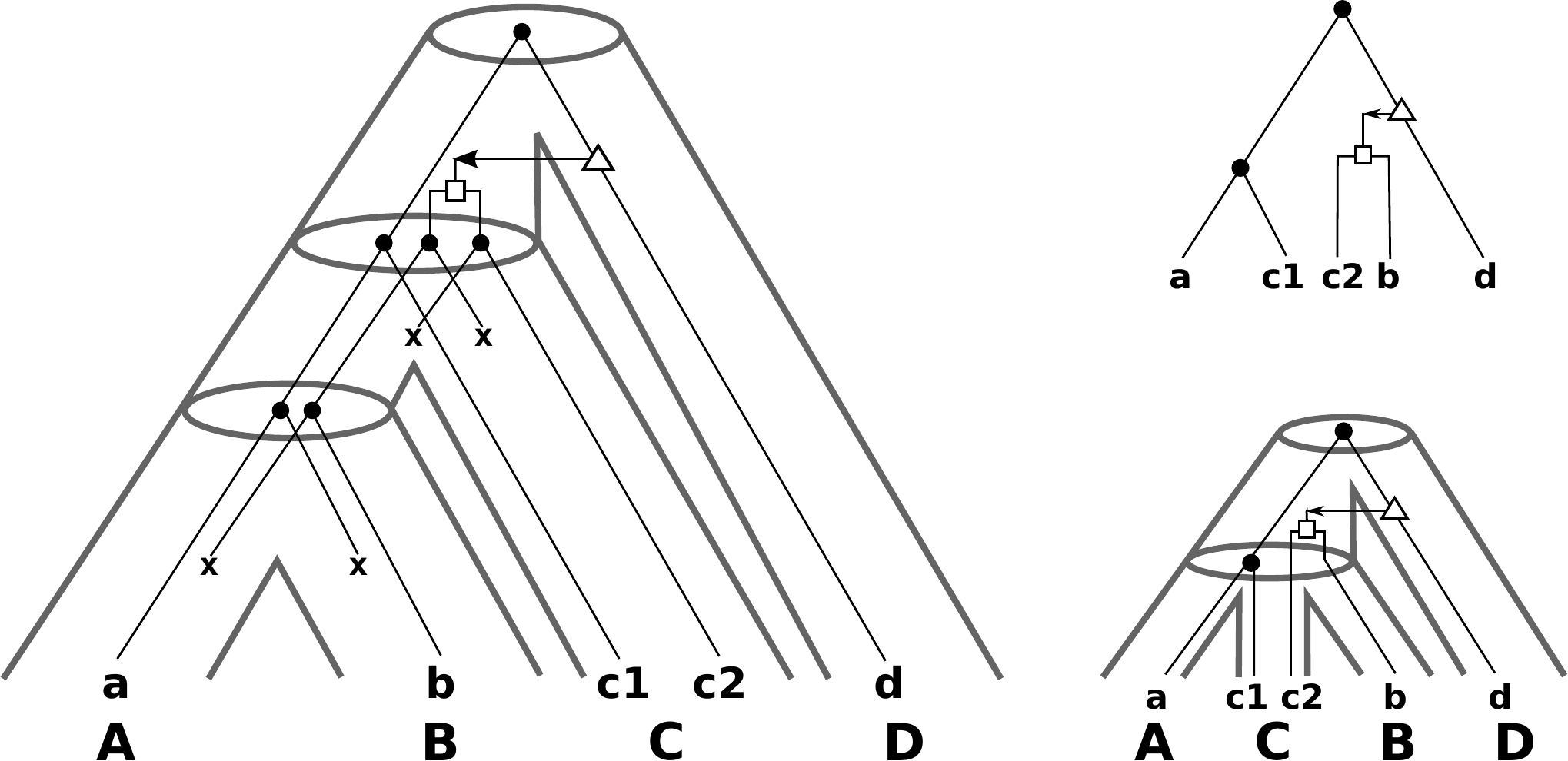}
  \end{center}
	\caption{\emph{Left:} 
		An example of a ``true'' history of a gene tree that evolves along the
		(tube-like) species tree (taken from \cite{Hellmuth:15a}). The set of
		extant genes $\Gen$ comprises $a,b,c_1,c_2$ and $d$ and $\sigma$ maps each
		gene in $\Gen$ to the species (capitals below the genes) $A,B,C,D\in
		\Spe$.	 
		\emph{Upper Right:}  
		The observable gene tree $(T;t,\sigma)$ is shown. To derive $\mc{S}(T;t,\sigma)$ 	
		we cannot use the triples  $\mc{R}_0(T)$, that is, we need to remove the transfer
		edges. To be more precise, 	
		 if we would consider $\mc{R}_0(T)$
		we obtain the triples $\rt{ac_1|d}$ and $\rt{c_2d|a}$ which leads to the
		two contradicting species triples $\rt{AC|D}$ and $\rt{CD|A}$. Thus, we
		restrict $\mc{R}_0$ to $\Th$ and obtain $\mc{R}_0(\Th) = \{\rt{ac_1|d}\}$.
		However, this triple alone would not provide enough information to obtain
		a species tree such that a valid reconciliation map $\mu$ can be
		constructed. Hence, we take $\mc{R}_1(\Th)=\{\rt{bc_2|d}\}$ into account
		and obtain $\mc{S}(T;t,\sigma) = \{\rt{AC|D},\rt{BC|D}\}$. 
		\emph{Lower Right:} 
				A least resolved species tree $S$ (obtained with \texttt{BUILD})
		that displays all triples in $\mc{S}(T;t,\sigma)$ together with the
		reconciled gene tree $(T;t,\sigma)$ is shown. Although $S$ does not
		display the triple $\rt{AB|C}$ as in the true history, this tree $S$ does
		not pretend a higher resolution than actually supported by $(T;t,\sigma)$.
		Clearly, as more gene trees (gene families) are available as more
		information about the resolution of the species tree can be provided. 
    }
	\label{fig:reconc}
\end{figure*}

\subsection{Binary Gene Trees}


In this section, we will be concerned
only with binary, i.e., ``fully resolved'' gene trees,
if not stated differently. 
This is justified by
the fact that, generically, a speciation or duplication event
instantaneously generates exactly two offspring. However, we will allow
also non-binary species tree to model incomplete knowledge of the exact
species phylogeny. Non-binary gene trees are discussed in Section \ref{sec:restMu}.

Hernandez et al.\  \cite{HHH+12} established the following characterization for the HGT-free case. 
\begin{theorem}
	For a given gene tree $(T;t, \sigma)$ on $\Gen$ that \emph{does not} contain
	HGT and $\mathfrak{S}\coloneqq \{\rt{\sigma(a)\sigma(b)|\sigma(c)} \colon
	\rt{ab|c} \in \mathcal{R}_0(T)\}$, the following statement is satisfied:
		
	There is a species tree on $\Spe = \sigma(\Gen)$ for $(T;t, \sigma)$
	 if	and only if the triple set $\mathfrak{S}$ is consistent.
\label{thm:charact-old}
\end{theorem}
We emphasize that the results established in \cite{HHH+12} are only valid
for binary gene trees, although this was not explicitly stated.
For an example that shows that Theorem \ref{thm:charact-old} 
does is no always satisfied for non-binary gene trees see 
Figure \ref{fig:counterBinary}. 
Lafond and El-Mabrouk \cite{Lafond2014,LDEM:16} established
a similar result as in Theorem \ref{thm:charact-old} by using
only species triples that can be obtained directly
from a given orthology/paralogy-relation. However, they require a stronger
version of axiom (O3.a), that is, the images of all children of a
speciation vertex must be pairwisely incomparable in the species tree. 
We, too, will use this restriction in Section \ref{sec:restMu}

In what follows, we generalize the latter result and show that consistency of
$\mathcal{S}(T;t,\sigma)$ characterizes whether there is a species tree $S$ for
$(T;t,\sigma)$ even if $(T;t,\sigma)$ contains HGT.

\begin{lemma}
	If $\mu$ is a reconciliation map from a gene tree $(T;t,\sigma)$ to 
	a species tree $S$
	and $\rt{ab|c} \in\mathcal{R}(T;t,\sigma)$, then 
	$\rt{\sigma(a)\sigma(b)|\sigma(c)}$ is displayed in $S$.
\label{lem:mu-triple}  
\end{lemma}
\begin{proof}
	Recall that $\Gen$ is the leaf set of $T=(V,E)$ and, by Lemma \ref{lem:part_gen}, 
	of $\Th$. 	
	Let $\{a,b,c\} \in
	\binom{\Gen}{3}$ and assume w.l.o.g.\ $\rt{ab|c} \in
	\mathcal{R}(T;t,\sigma)$.

	\smallskip
	First assume that $\rt{ab|c} \in \mc{R}_0$, that is $\rt{ab|c}$ is displayed
	in $\Th$ and $t(\lca_{\Th}(a,b,c)) = \bul$. For simplicity set
	$u=\lca_{\Th}(a,b,c)$ and let $x,y$ be its children in $\Th$. 
	Since $\rt{ab|c} \in \mc{R}_0$, we can assume that w.l.o.g.\ 
	$a,b\in L_{\Th}(x)$ and $c\in L_{\Th}(y)$. 	
	Hence,  $x\succeq_{\Th}	\lca_{\Th}(a,b)$ and $y\succeq_{\Th} c$. 
   Condition (M3) implies that
	$\mu(y)\succeq_S \mu(c) = \sigma(c)$. Moreover, Condition (M3) and Lemma
	\ref{lem:cond-mu}(\ref{item:lca}) imply that $\mu(x)\succeq_S
	\mu(\lca_{\Th}(a,b)) \succeq_S \lca_S(\mu(a),\mu(b)) =
	\lca_S(\sigma(a),\sigma(b))$. 
Since $t(u)=\bul$, we can apply 	Lemma \ref{lem:cond-mu}(\ref{item:speci}) 
	and conclude that $\mu(x)$ and $\mu(y)$ are incomparable in $S$. Hence,
	$\sigma(c)$ and $\lca_S(\sigma(a),\sigma(b))$ are incomparable. Thus, the
	triple $\rt{\sigma(a)\sigma(b)|\sigma(c)}$ must be displayed in $S$.

	Now assume that $\rt{ab|c} \in \mc{R}_i$ for some transfer edge $e_i = (x,y)\in\EHGT$. 
	For $e_i = (x,y)$ we either have $a,b\in L_{\Th}(x)$ and $c\in L_{\Th}(y)$ or $c\in
	L_{\Th}(x)$ and $a,b\in L_{\Th}(y)$. W.l.o.g.\ let 
	 $a,b\in L_{\Th}(x)$ and $c\in L_{\Th}(y)$. Thus, $x\succeq_{\Th}
	\lca_{\Th}(a,b)$ and $y\succeq_{\Th} c$. Condition (M3) implies that
	$\mu(y)\succeq_S \mu(c) = \sigma(c)$. Moreover, Condition (M3) and Lemma
	\ref{lem:cond-mu}(\ref{item:lca}) imply that $\mu(x)\succeq_S
	\mu(\lca_{\Th}(a,b)) \succeq_S \lca_S(\mu(a),\mu(b)) =
	\lca_S(\sigma(a),\sigma(b))$. Since $t(x)=\tri$, we can apply (M2.iii)
	and conclude that $\mu(x)$ and $\mu(y)$ are incomparable in $S$. Hence,
	$\sigma(c)$ and $\lca_S(\sigma(a),\sigma(b))$ are incomparable. Thus, the
	triple $\rt{\sigma(a)\sigma(b)|\sigma(c)}$ must be displayed in $S$.
\hfill \end{proof}

\begin{lemma}	
	Let $S=(W,F)$ be a species tree on $\Spe$. 
	Then there is reconciliation map 	$\mu$  from a  gene tree $(T;t,\sigma)$ to $S$
	whenever $S$ displays all triples in $\mathcal{S}(T;t,\sigma)$.
\label{lem:triple-mu}
\end{lemma}
\begin{proof}
	Recall that $\Gen$ is the leaf set of $T=(V,E)$ and, by Lemma \ref{lem:part_gen}, of $\Th$. In what follows, we write $\LT(u)$ instead of
	the more complicated writing $L_{\Th}(u)$ and, for consistency and simplicity, 
	we also often write $\sigma(\LT(u))$ instead of $\sT(u)$. 
	Put $S=(W,F)$ and $\mathcal{S} =
	\mathcal{S}(T;t,\sigma)$. We first consider the subset $U=\{x\in V \mid x\in \Gen
	\text{ or } t(x) =  \bul\}\}$  
   of $V$ comprising the leaves and speciation
	vertices of $T$.

	In what follows we will explicitly construct $\mu: V \to W\cup F$
	and verify that $\mu$ satisfies Conditions (M1), (M2) and (M3). 
	To this end, we first set for all $x\in U$:
	\begin{enumerate}
		\item[(S1)] $\mu(x) = \sigma(x)$, if $x\in \Gen$,
		\item[(S2)] $\mu(x)= \lca_S(\sigma(\LT(x)))$, if $t(x)=\bul$.
	\end{enumerate}
	Conditions (S1) and (M1), as well as (S2) and (M2.i) are equivalent. 
	
	For later reference, we show that $\lca_S(\sigma(\LT(x))) \in W^0 =
	W\setminus \Spe$ and that there are two leaves $a,b\in \LT(x)$ such that 
	 $\sigma(a) \neq \sigma(b)$,  whenever $t(x)=\bul$.
	Condition (O3.a) implies that
	$x$ has  two children $v$ and
	$w$ in $T$ such that  $\sigma(\LT(v)) \cap
	\sigma(\LT(w)) = \emptyset$. Moreover, 
	Lemma \ref{lem:part_gen} implies that both $\LT(v)$ and $\LT(w)$ are
	non-empty subsets of $\Gen$ and hence, neither 
	$\sigma(\LT(v))=\emptyset$ nor 	$\sigma(\LT(w))=\emptyset$.
	Thus, there are two leaves $a, b\in \LT(x)$ such
	that $\sigma(a) \neq \sigma(b)$. Hence, $\lca_S(\sigma(\LT(x))) \in W^0 =
	W\setminus \Spe$. 

	\smallskip
	\begin{description}
	\item \textbf{Claim 1:} \emph{For all $x,y\in U$ with $x\prec_{\Th} y$ we have $\mu(x)\prec_S \mu(y)$.}\\
	Note, $y$ must be an interior vertex, since $x\prec_{\Th} y$. Hence $t(y)=\bul$.
	
	If $x$ is a leaf, then $\mu(x)=\sigma(x)\in \Spe$. As argued above, 
	$\mu(y) \in W\setminus \Spe$. Since $x\in \LT(y)$ and 
	$\sigma(\LT(y))\neq \emptyset$, 
	 we have 
	$\sigma(x) \in \sigma(\LT(y))\subseteq \Spe$ and thus, $\mu(x)\prec_S \mu(y)$. 

	Now assume that $x$ is an interior vertex and hence, $t(x)=\bul$.
	Again, there are leaves $a,b \in \LT(x)$ with 
	$A = \sigma(a)\neq \sigma(b)=B$.
	Since $t(y)=\bul$, vertex $y$ has two children in $\Th$. Let $y'$ denote the child of  
	$y$ with	$x\preceq_{\Th} y'$. 
	Since $\LT(x)\subseteq \LT(y')\subsetneq \LT(y)$,
	we have $\LT(y)\setminus\LT(y')\neq \emptyset$ and, by 
	Condition (O3.a), 
	there is a gene $c\in \LT(y)\setminus \LT(y') \subseteq \LT(y)\setminus \LT(x)$ with
	$\sigma(c)=C\neq A,B$. 
	By construction, $\rt{ab|c}\in \mc{R}_0$ and hence,
	$\rt{AB|C}\in \mathcal{S}(T;t,\sigma)$.
	Hence, $\lca_S(A,B)\prec_S \lca_S(A,B,C)$. 
	Since this holds for all triples $\rt{x'x''|z}$ with $x',x''\in \LT(x)$
	and $z\in \LT(y)\setminus \LT(y')$, %
	we can conclude that 
	\begin{align*}
		\mu(x) &= \lca_S(\sigma(\LT(x)))  \\ &\prec_S        
				 \lca_S(\sigma(\LT(x))\cup \sigma(\LT(y) \setminus \LT(y'))). 
	\end{align*}
	Since $\sigma(\LT(x))\cup \sigma(\LT(y) \setminus \LT(y')) \subseteq \sigma(\LT(y))$
   we obtain 
	\begin{align*}		
		    &\lca_S(\sigma(\LT(x))\cup \sigma(\LT(y) \setminus \LT(y'))) \\ &\preceq_S  
					\lca_S(\sigma(\LT(y))) = \mu(y). 
	\end{align*}
	Hence,  $\mu(x)\prec_S\mu(y)$.

	\hfill $_{\textrm{\footnotesize{-- End Proof Claim 1 -- }}}$
	\end{description}	\smallskip

	We continue to extend $\mu$ to the entire set $V$. To this end, observe first
	that if $t(x) \in \{\tri, \squ\}$ then we wish to map $x$ on an edge
	$\mu(x) = (u,v) \in F$ such that Lemma \ref{lem:cond-mu-lca}
	is satisfied: $v\succeq_S \lca_S(\sigma(\LT(x)))$. Such an edge exists for $v
	= \lca_S(\sigma(\LT(x)))$ in $S$ by construction. Every speciation vertex $y$
	with $y\succ_{\Th} x$ therefore necessarily maps on the vertex $u$ or above,
	i.e., $\mu(y) \succeq_S u$ must hold.
	Thus, we set:
	\begin{enumerate}
		\item[(S3)] $\mu(x) = (u,\lca_S(\sigma(\LT(x))))$, if $t(x)\in \{\tri, \squ\}$,
	\end{enumerate}
	which now makes $\mu$ a map from $V$ to $W\cup F$.	

	By construction of $\mu$, Conditions (M1), (M2.i), (M2.ii) are satisfied by $\mu$. 

	We proceed to show that (M3) is satisfied.

	\smallskip
	\begin{description}
	\item \textbf{Claim 2:} 
	\emph{ For all $x,y\in V$ with $x\prec_{\Th} y$, Condition (M3) is satisfied.}\\ 
	If both $x$ and $y$ are speciation vertices, then we can apply the Claim 1
	to conclude that $\mu(x)\prec_S \mu(y)$.
	If $x$ is a leaf, then we argue similarly as in the proof of Claim 1
	to conclude that $\mu(x)\preceq_S \mu(y)$.

	Now assume that both $x$ and $y$ are interior vertices of $T$ and 
	at least one vertex of $x,y$ is not a speciation vertex.
	Since, $x\prec_{\Th} y$ we have $\LT(x) \subseteq \LT(y)$ and thus,  $\sigma(\LT(x)) \subseteq
	\sigma(\LT(y))$.

	We start with the case $t(y)=\bul$ and $t(x)\in \{\squ, \tri\}$.
	Since $t(y)=\bul$, vertex $y$ has two children in $\Th$. Let $y'$ be the child of $y$ with
   $x\preceq_{\Th} y'$. 
	If $\sigma(\LT(x))$ contains only one species $A$, then
	$\mu(x) = (u,A)\prec_S u\preceq_S \lca_S(\sigma(\LT(y))) = \mu(y)$.
	If $\sigma(\LT(x))$ contains at least two species, then there are $a,b\in
	\LT(x)$ with $\sigma(a)=A\neq \sigma(b)=B$
	Moreover, since $\LT(x)\subseteq \LT(y')\subsetneq \LT(y)$,
	we have $\LT(y)\setminus\LT(y')\neq \emptyset$ and, by 
	Condition (O3.a), 
	there is a gene $c\in \LT(y)\setminus \LT(y') \subseteq \LT(y)\setminus \LT(x)$ with
	$\sigma(c)=C\neq A,B$. 
	By construction, $\rt{ab|c}\in \mc{R}_0$ and hence
	$\rt{AB|C}\in \mathcal{S}(T;t,\sigma)$. Now we can argue similar as in the proof of the
	Claim 1, to see that 
	 \begin{align*}
	\mu(x) &= (u,\lca_S(\sigma(\LT(x)))) \prec_S u \\ &\preceq_S
	\lca_S(\sigma(\LT(y))) = \mu(y).
	 \end{align*}

	If $t(x)=\bul$ and $t(y)\in \{\squ, \tri\}$, then $\sigma(\LT(x))
	\subseteq \sigma(\LT(y))$ implies that 
	 \begin{align*}
	\mu(x) &=	\lca_S(\sigma(\LT(x)))\preceq_S \lca_S(\sigma(\LT(y))) \\
				&\prec_S(u,\lca_S(\sigma(\LT(y)))) = \mu(y).
	 \end{align*}

	Finally assume that $t(x),t(y)\in \{\squ, \tri\}$. If $\sigma(\LT(x))
	= \sigma(\LT(y))$, then $\mu(x) = \mu(y)$. Now let $\sigma(\LT(x)) \subsetneq
	\sigma(\LT(y))$ which implies that $\lca_S(\sigma(\LT(x)))\preceq_S
	\lca_S(\sigma(\LT(y)))$. If $\lca_S(\sigma(\LT(x))) =
	\lca_S(\sigma(\LT(y)))$, then $\mu(x) = \mu(y)$. If
	$\lca_S(\sigma(\LT(x)))\prec_S \lca_S(\sigma(\LT(y)))$, 
	then 
	 \begin{align*}
		\mu(x)&=(u,\lca_S(\sigma(\LT(x)))) \prec_S u \\
				&\preceq_S \lca_S(\sigma(\LT(y)))
				 \prec (u',\lca_S(\sigma(\LT(y)))) \\ &=\mu(y). 
	 \end{align*} 	\hfill $_{\textrm{\footnotesize{-- End Proof Claim 2 --}}}$
	\end{description}
	\smallskip

	It remains to show (M2.iii), that is, if $e_i=(x,y)$ is a transfer-edge, then
	$\mu(x)$ and $\mu(y)$ are incomparable in $S$. Since $(x,y)$ is a transfer
	edge and by Condition (O3.b), $\sigma(\LT(x)) \cap \sigma(\LT(y)) = \emptyset$.
	If $\sigma(\LT(x))=\{A\}$ and $\sigma(\LT(y))=\{C\}$, then $\mu(x) = (u,A)$
	and $\mu(y) = (u',C)$. Since $A$ and $C$ are distinct leaves in $S$, $\mu(x)$
	and $\mu(y)$ are incomparable.
	Assume that $|\sigma(\LT(x))|>1$. Hence, 
	there are leaves $a,b \in \LT(x)$ with $A = \sigma(a)\neq \sigma(b)=B$
	and $c\in \LT(y)$ with $\sigma(c)=C\neq A,B$. By construction, $\rt{ab|c}\in
	\mc{R}_i$ and hence, $\rt{AB|C}\in \mathcal{S}(T;t,\sigma)$.
	The latter is fulfilled for all triples $\rt{x'x''|c}\in \mc{R}_i$ with $x',x''\in \LT(x)$, 
	and, therefore, $\lca_S(\sigma(\LT(x))\cup\{C\}) \succ_S \lca_S(\sigma(\LT(x)))$.
	Set $v=\lca_S(\sigma(\LT(x))\cup \{C\})$.
	Thus, there is an edge $(v,v')$ in $S$
	with $v'\succeq_S \lca_S(\sigma(\LT(x)))$ and an edge $(v,v'')$  such that $v''\succeq_S C$. 
	Hence, either $\mu(x) = (v,v')$ or $\mu(x) = (u,\lca_S(\sigma(\LT(x)))$ and $v'\succeq_S u$. 
	Assume that $\sigma(\LT(y))$ contains only the species $C$ 
	and thus,  $\mu(y) = (u',C)$. 
	Since $v''\succeq_S C$, we have either 
	$v'' =  C$ which implies that $\mu(y) = (v,v'')$ or 
	$v'' \succ_S  C$ which implies that $\mu(y) = (u',C)$ and $v''\succeq_S u'$. 
	Since both vertices $v'$ and $v''$ are incomparable in $S$, 
	so $\mu(x)$ and $\mu(y)$ are.
  	If $|\sigma(\LT(y))|>1$, then we set  $v=\lca_S(\sigma(\LT(x))\cup \sigma(\LT(y)))$
	and we can argue analogously as above and conclude that there are edges 
	$(v,v')$ and $(v,v'')$ in $S$ such that 
	$v'\succeq_S \lca_S(\sigma(\LT(x)))$  and 
	$v''\succeq_S \lca_S(\sigma(\LT(y)))$.
	Again, 
	since $v'$ and $v''$ are incomparable in $S$ and by construction of $\mu$, 
	$\mu(x)$ and $\mu(y)$ are incomparable.
\hfill \end{proof}

Lemma \ref{lem:mu-triple} implies that consistency of the triple set
$\mathcal{S}(T; t,\sigma )$ is necessary for the existence of a reconciliation
map from $(T; t,\sigma )$ to a species tree on $\Spe$. Lemma
\ref{lem:triple-mu}, on the other hand, establishes that this is also
sufficient. Thus, we have
\begin{theorem}
	There is a species tree on $\Spe = \sigma(\Gen)$ for a  gene tree 
	$(T;t, \sigma)$ on $\Gen$ if
	and only if the triple set $\mathcal{S}(T; t,\sigma )$ is consistent.
\label{thm:charact}
\end{theorem}

\subsection{Non-Binary Gene Trees}
\label{sec:restMu}

\begin{figure*}[ht]
  \begin{center}
    \includegraphics[width=.8\textwidth]{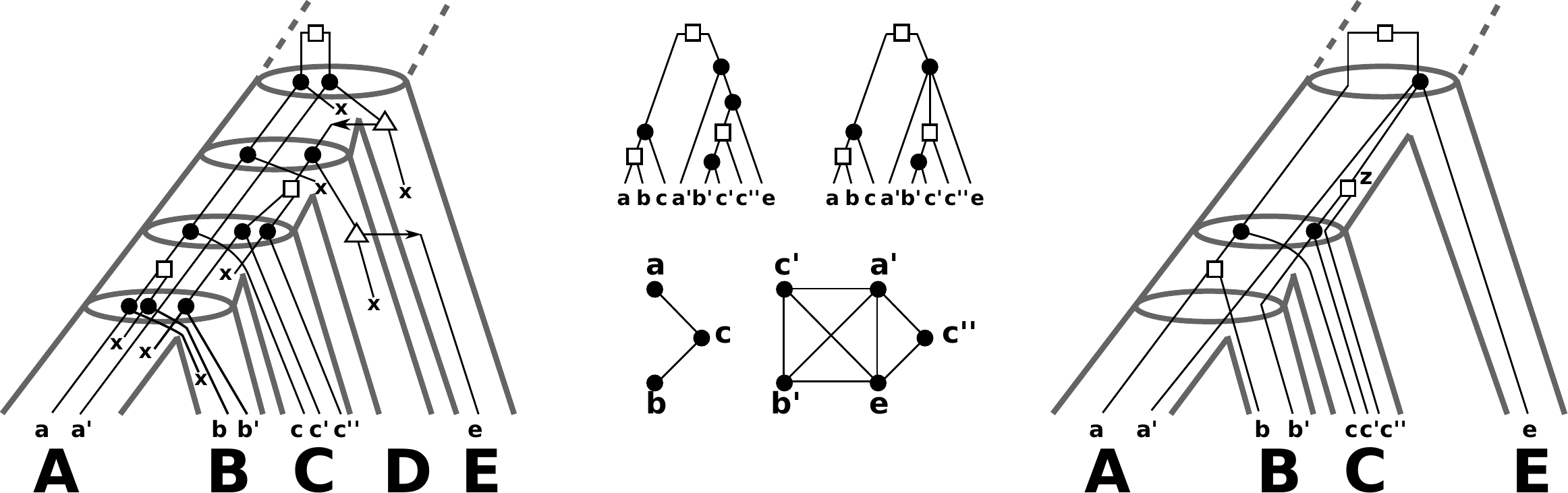}
  \end{center}
	\caption{
	Consider the ``true'' history (left) that is also shown in Figure \ref{fig:true}.
	The center-left gene tree   $(T;t,\sigma)$  is biologically feasible and 
	obtained as the observable part of the true history.
	There is no reconciliation map for  $(T;t,\sigma)$  to  any species 
	tree according to Def.\ \ref{def:mu} because $\mc{S}(T;t,\sigma)$ is inconsistent (cf.\ Thm.\ \ref{thm:charact}). 
	The graph in the lower-center 
	depicts the orthology-relation  that comprises all pairs $(x,y)$ of 
	vertices for which $t(\lca(x,y)) =\bul$. 
	The center-right gene tree  $(T';t,\sigma)$ is non-binary and can directly be computed 
	from the orthology-relation. Although $\mc{S}(T';t,\sigma)$ is inconsistent, 
	there is a valid reconciliation map $\mu$ to a species tree for $(T';t,\sigma)$ 
	according to Def.\ \ref{def:mu} (right). Note,  both trees $(T;t,\sigma)$ and 
	$(T';t,\sigma)$ satisfy axioms (O1)-(O3) and even (O3.A). However, the reconciliation 	
	map $\mu$ does not
	satisfy the extra condition (M2.iv), since $\mu(z)$ and $\mu(a')=A$ are comparable, 
	although $z$ and $a'$ are children of a common speciation vertex. 
	Therefore, Axioms (O1)-(O3) and (O3.A) do not imply (M2.iv).
	Moreover, Thm.\ \ref{thm:charact2} implies that there is no restricted reconciliation map
	for  $(T;t,\sigma)$  as well as $(T';t,\sigma)$ and any species tree, since $\mc{S}(T;t,\sigma)$ 
	and  $\mc{S}(T';t,\sigma)$ are inconsistent. See text for further details.
 }
	\label{fig:counterBinary}
\end{figure*}

Now, we consider arbitrary, possibly non-binary gene trees that might be used  
to model incomplete knowledge of the exact genes phylogeny. 
Consider the ``true'' history of a gene tree that evolves along the (tube-like)
species tree in Figure \ref{fig:counterBinary} (left).
The observable gene tree $(T;t,\sigma)$ is shown in  
\ref{fig:counterBinary} (center-left). Since 
 $\rt{ab|c},\rt{b'c'|a'} \in \mc{R}_0$, 
 we obtain a set of  
species triples $\mathcal{S}(T;t,\sigma)$ that contain the pair of 
inconsistent species triple $\rt{AB|C},\rt{BC|A}$. Thus, 
there is no reconciliation map for 
$(T;t,\sigma)$ and any species tree, 
although $(T;t,\sigma)$ is biologically feasible. 
Consider now 
the ``orthology'' graph $G$ (shown below the gene trees) that 
has as vertex set $\Gen$ and two genes $x,y$ are connected by an edge if $\lca(x,y)$ is a 
speciation vertex. Such graphs can be obtained from orthology
inference methods \cite{Lechner:11a,Lechner:14,inparanoid:10,TG+00} 
and the corresponding \emph{non-binary} 
gene tree  $(T';t,\sigma)$ (center-right) is constructed from such estimates
 (see \cite{HHH+13,HSW:16,HW:16b} for further
    details).
Still, we can see that  $\mathcal{S}(T';t,\sigma)$ contains the two 
inconsistent species triples $\rt{AB|C},\rt{BC|A}$. However, there is 
a reconciliation map $\mu$ according to 
Definition \ref{def:mu} and a species tree $S$, as shown in 
Figure \ref{fig:counterBinary} (right). 
Thus, consistency of  $\mathcal{S}(T';t,\sigma)$  does not characterize whether there
is a valid reconciliation map for non-binary gene trees.

In order to obtain a similar result as in Theorem \ref{thm:charact} for non-binary 
gene trees we have to strengthen observability axiom (O3.a) to \smallskip
\begin{description}
\item[(O3.A)] 
		If $x$ is a speciation vertex with children $v_1,\dots,v_k$, then 
		$\sT(v_i) \cap \sT(v_j) =\emptyset$,  $1\leq i<j\leq k$; \smallskip
\end{description}
and to add an extra event constraint to Definition \ref{def:mu}: \smallskip
\begin{description}
	\item[(M2.iv)]  Let $v_1,\dots,v_k$ be the children of the speciation vertex $x$. 
						 Then, $\mu(v_i)$ and $\mu(v_j)$ are incomparable in $S$, $1\leq i<j\leq k$.
\end{description} \smallskip

We call a reconciliation map that additionally satisfies (M2.iv) a  \emph{restricted reconciliation map}.
Such restricted reconciliation maps 
satisfy the condition as required in \cite{LDEM:16,Lafond2014} for the HGT-free case.  
It can be shown that restricted reconciliation maps imply Condition (O3.A), 
however, the converse is not true in general, see Figure \ref{fig:counterBinary}.
Hence, we cannot use the axioms (O1)-(O3) and (O3.A)
to derive Condition (M2.iv)  - similar to Lemma \ref{lem:cond-mu}(\ref{item:speci}) 
- and thus, need to claim it.

It is now straightforward to obtain the next result. 
\begin{lemma}
	If $\mu$ is a restricted reconciliation map 
   from $(T;t,\sigma)$ to $S$
	and $\rt{ab|c} \in\mathcal{R}(T;t,\sigma)$, then 
	$\rt{\sigma(a)\sigma(b)|\sigma(c)}$ is displayed in $S$.
\label{lem:mu-triple-mB}  
\end{lemma}
\begin{proof}
	Let $\{a,b,c\} \in
	\binom{\Gen}{3}$ and assume w.l.o.g.\ $\rt{ab|c} \in
	\mathcal{R}(T;t,\sigma)$.

	\smallskip
	First assume that $\rt{ab|c} \in \mc{R}_0$, that is $\rt{ab|c}$ is displayed
	in $\Th$ and $t(\lca_{\Th}(a,b,c)) = \bul$. For simplicity set
	$u=\lca_{\Th}(a,b,c)$. Hence, there are two children  $x,y$ of $u$ in $\Th$
	such that  w.l.o.g.\ $a,b\in L_{\Th}(x)$ and $c\in L_{\Th}(y)$. 	
	Now we can argue analogously as in the proof of Lemma \ref{lem:mu-triple}  
	after replacing ``we can apply 	Lemma \ref{lem:cond-mu}(\ref{item:speci}) ''
	by ``we can apply Condition (M2.iv)''. 
	The proof for $\rt{ab|c} \in \mc{R}_i$ remains the same as in  Lemma \ref{lem:mu-triple}.
\end{proof}


\begin{lemma}	
	Let $S$ be a species tree on $\Spe$. 
	Then, there is a restricted reconciliation map 	$\mu$  from a  gene tree $(T;t,\sigma)$ 
   that satisfies also (O3.A) to $S$
	whenever $S$ displays all triples in $\mathcal{S}(T;t,\sigma)$.
\label{lem:triple-mu2}
\end{lemma}
\begin{proof}
	The proof is similar to the proof of Lemma \ref{lem:triple-mu2}. 
	However, note that a speciation vertex might have more than two children. 
	In these cases, one simply has to apply
	 Axiom (O3.A) instead of Lemma (O3.a)
	to conclude that (M1),(M2.i)-(M2.iii), (M3) are
	satisfied. 

	It remains to show that 	(M2.iv) is satisfied. 
	To this end, let $x$ be a speciation vertex in $T$ and the set of its children
	$C(x) = \{v_1,\dots,v_k\}$. 
	By axiom (O3.A) we have	$\sT(v_i) \cap \sT(v_j) =\emptyset$ for all $i\neq j$.
	Consider the following partition of $C(x)$ into $C_1$ and $C_2$ that contain
    all vertices $v_i$ with $|\sT(v_i)|=1$ and $|\sT(v_i)|>1$, respectively. 
	By construction of $\mu$, for all vertices in $v_i,v_j\in C_1$, $i\neq j$
	we have that $\mu(v_i)\in \{\sigma(v_i), (u,\sigma(v_i)) \}$ and 
	$\mu(v_j)\in \{\sigma(v_j), (u',\sigma(v_j)) \}$  are incomparable. 
	Now let  $v_i\in C_1$ and  $v_j\in C_2$. Thus there are $A,B\in \sT(v_j)$
	and $\sigma(v_i)=C$. Hence, $\rt{AB|C} \in  \mathcal{S}(T;t,\sigma)$
	Thus, $\lca_S(A,B)$	must be incomparable to $C$ in $S$. 
	Since the latter is satsfied for all species in $\sT(v_j)$, 	
    $\lca_S( \sT(v_j))$	and $C$ must be incomparable in $S$. 
	Again, by construction of $\mu$, we see that 
	$\mu(v_i)\in \{C, (u,C) \}$ and 
	$\mu(v_j)\in \{\lca_S( \sT(v_j)), (u',\lca_S( \sT(v_j))) \}$
	are incomparable in $S$. 
	Analogously, if $v_i,v_j\in C_2$, $i\neq j$, then 
    all triples $\rt{AB|C}$ and $\rt{CD|A}$ for all $A,B\in \sT(v_j)$
	and $C,D\in \sT(v_j)$ are contained in $\mathcal{S}(T;t,\sigma)$
	and thus, displayed by $S$. 
	Hence, $\lca_S( \sT(v_i))$	 and $\lca_S( \sT(v_j))$	
	must be incomparable in $S$. Again, by construction of $\mu$, 
	we obtain that $\mu(v_i)\in \{\lca_S( \sT(v_i)), (u,\lca_S( \sT(v_i))) \}$ and 
	$\mu(v_j)\in \{\lca_S( \sT(v_j)), (u',\lca_S( \sT(v_j))) \}$
	are incomparable in $S$. Therefore, (M2.iv) is satisfied. 
\end{proof}

As in the binary case, we obtain
\begin{theorem}
	There is a restricted reconciliation map for 
	a gene tree $(T;t, \sigma)$ on $\Gen$ that satisfies also (O3.A) and some species tree on
	$\Spe = \sigma(\Gen)$ if
	and only if the triple set $\mathcal{S}(T; t,\sigma )$ is consistent.
\label{thm:charact2}
\end{theorem}
	

\subsection{Algorithm}

The proof of Lemma \ref{lem:triple-mu} and \ref{lem:triple-mu2} is constructive and we summarize the
latter findings in Algorithm \ref{alg:alg}, see Figure \ref{fig:reconc} for an
illustrative example. 
 
\begin{lemma}
		 Algorithm \ref{alg:alg} returns a species tree $S$ for a binary gene tree $(T;t,\sigma)$ 
		 and a reconciliation map $\mu$ in polynomial time, if one exists
		 and otherwise, returns that there is no species tree for $(T;t,\sigma)$.

		If $(T;t,\sigma)$ is non-binary but satisfies Condition (O3.A), 
	   then Algorithm \ref{alg:alg} returns a species tree $S$ for $(T;t,\sigma)$ 
		 and a  restricted reconciliation map $\mu$ in polynomial time, if one exists
		 and otherwise, returns that there is no species tree for $(T;t,\sigma)$.
\end{lemma}
\begin{proof}
	Theorem  \ref{thm:charact} and the construction of $\mu$ in the proof of Lemma
	\ref{lem:triple-mu} and \ref{lem:triple-mu2} implies the correctness of the algorithm. 
	
	For the runtime observe that all tasks, computing $\mathcal{S}(T;t,\sigma)$, using the 
	\texttt{BUILD} algorithm \cite{Aho:81,Semple:book} and  the construction of the map
	$\mu$ \cite[Cor.7]{HHH+12} can be done  in polynomial time. 
\hfill  
\end{proof}

\begin{figure*}[tbp]
	\begin{center}
		\includegraphics[width=.8\textwidth]{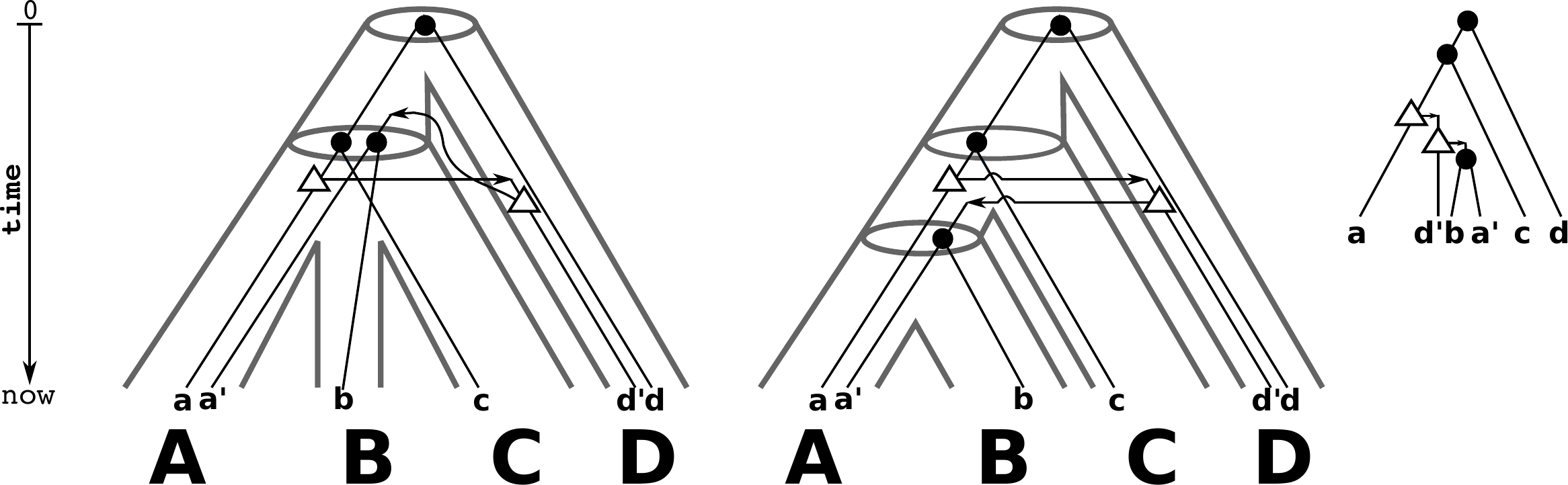}
	\end{center}
	\caption{From the binary gene tree $(T;t,\sigma)$  (right)                                             
		we obtain the species triples $\mathcal{S}(T;t,\sigma) = \{\rt{AB|D},\rt{AC|D}\}$.                  
		Shown are two  (tube-like)  species trees  (left and middle) that 
		display $\mathcal{S}(T;t,\sigma)$. 	The respective  reconciliation maps 
		for $T$ and $S$ are given implicitly by 
		drawing $T$ within the species tree $S$. 
		The left tree $S$ is least resolved for $\mathcal{S}(T;t,\sigma)$. 
		Although there is even a unique reconciliation map  from $T$ to $S$,  
		this map is not time-consistent. Thus, no time-consistent reconciliation
		between $T$ and $S$ exists. 
		On the other hand, for $T$ and the middle species tree $S'$ (that is a refinement of $S$)
		there is a time-consistent reconciliation map. 
		Figure \ref{fig:reconc} provides an example that shows that also least-resolved species trees 
		can have a 	time-consistent reconciliation map with gene trees.
	}
	\label{fig:least}
\end{figure*}

In our examples, the species trees that display $\mathcal{S}(T; t,\sigma )$ is
computed using the $O(|L_R||R|)$ time algorithm \texttt{BUILD}, that either
constructs a tree $S$ that displays all triples in a given triple set $R$ or
recognizes that $R$ is not consistent. However, any other supertree method might
be conceivable, see \cite{Bininda:book} for an overview. The tree $T$ returned
by $\texttt{BUILD}$ is least resolved, i.e., if $T'$ is obtained from $T$ by
contracting an edge, then $T'$ does not display $R$ anymore. However, the trees
generated by $\texttt{BUILD}$ do not necessarily have the minimum number of
internal vertices, i.e., the trees may resolve multifurcations in an arbitrary
way that is not implied by any of the triples in $R$. Thus, depending on $R$,
not all trees consistent with $R$ can be obtained from $\texttt{BUILD}$. 
Nevertheless, in \cite[Prop.\ 2(SI)]{Hellmuth:15a} the following result was
established. 
\begin{lemma}
	Let $R$ be a consistent triple set.
	If the tree $T$ obtained with $\texttt{BUILD}$ applied on $R$ is binary, 
	then $T$ is a unique tree on $L_R$ that displays $R$, i.e., for any tree $T'$ on $L_R$ 
	that displays  $R$ we have $T'\simeq T$.
	\label{lem:build-unique}
\end{lemma}

\begin{algorithm}[tbp]
\caption{\texttt{ReconcT}}
\label{alg:Q}
\begin{algorithmic}[1]
  	\Require Non-binary gene tree that satisfies (O3.A) or  binary gene tree  $(T;t,\sigma)$ on $\Gen$
   \Ensure Species tree $S$ for $T$ and a (restricted) reconciliation map $\mu$, if one exists
	\State Compute $\mathcal{S}(T;t,\sigma)$;
	\If{\texttt{BUILD} recognizes $\mathcal{S}(T;t,\sigma)$ as not consistent} 
		\State {write} \emph{``There is no species tree for $(T;t,\sigma)$''} and \textbf{stop;}	
	\EndIf
		\State Let $S=(W,F)$ be the resulting species tree that displays $\mathcal{S}(T;t,\sigma)$ obtained with \texttt{BUILD};
		\For{all $x\in V(T)$}
			\If{$x\in \Gen$}{  set $\mu(x) = \sigma(x)$;			}
			\ElsIf{$t(x)=\bul$} set $\mu(x)=\lca_S(\sT(x))\in W^0$;
			\Else{} set $\mu(x)=(u,\lca_S(\sT(x)))\in F$;
			\EndIf
		\EndFor
		\State \Return $S$ and $\mu$;
\end{algorithmic}
\label{alg:alg}
\end{algorithm}

So-far, we have shown that event-labeled gene trees $(T;t,\sigma)$ for
which a species tree exists can be characterized by a set of species triples
$\mc{S}(T;t,\sigma)$ that is easily constructed from a subset of triples
displayed in $T$. From a biological point of view, however, it is necessary to
reconcile a gene tree with a species tree such that genes do not ``travel
through time''. In \cite{N+17}, the authors gave algorithms
to check whether a given
reconciliation map $\mu$ is time-consistent, and an algorithm with the same
time complexity for the construction of a time-consistent reconciliation
maps, provided one exists.
These algorithms require as input an event-labeled gene tree and species tree. 
Hence, a necessary condition for the existence of time-consistent reconciliation maps is given by
consistency of the species triple $\mc{S}(T;t,\sigma)$ derived from $(T;t,\sigma)$.
However, there are
possibly exponentially many species trees that are consistent with $\mc{S}(T;t,\sigma)$ for which some
of them have a time-consistent reconciliation map with $T$ and some not, see Figure \ref{fig:least}. The
question therefore arises as whether there is at least one species tree $S$ with time-consistent
map, and if so, construct $S$.


\section{Limitations of Informative Triples and Reconciliation Maps}

\begin{figure*}[ht]
  \begin{center}
    \includegraphics[width=.8\textwidth]{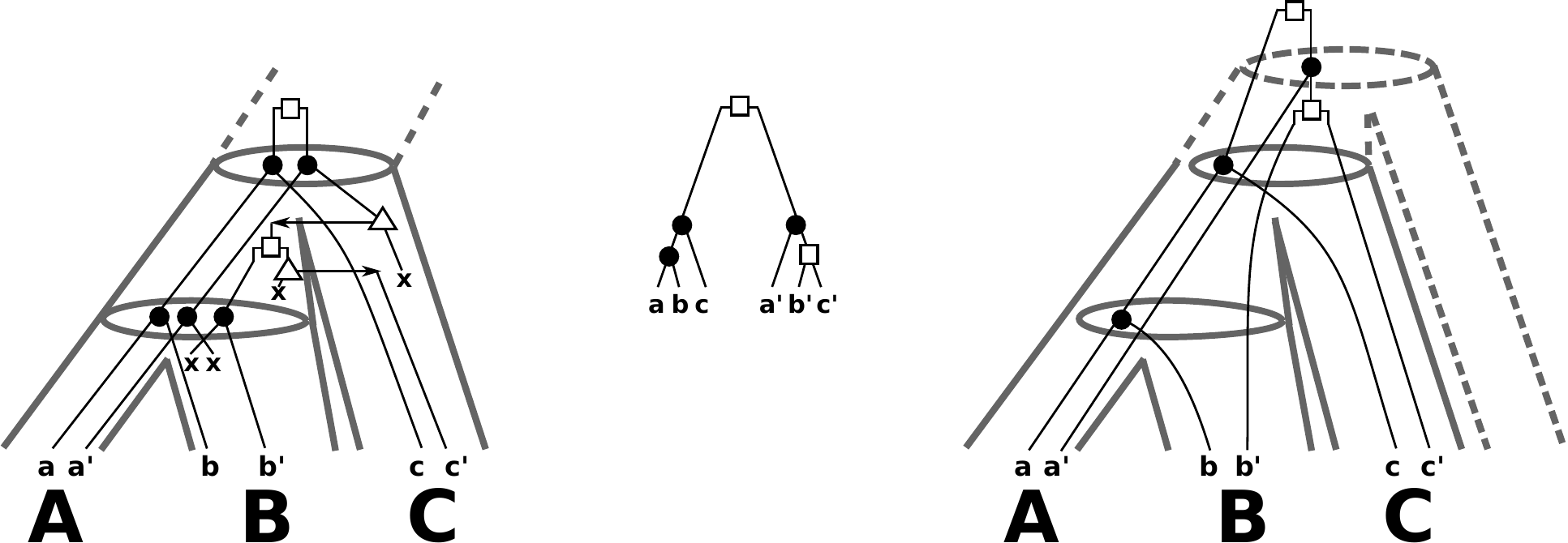}
  \end{center}
	\caption{Shown is a binary and biologically feasible gene tree $(T;t,\sigma)$ (center)
				that is obtained as the observable part of the true scenario (left). 
				However, there is no reconciliation map for $(T;t,\sigma)$ to  any species 
				tree according to Def.\ \ref{def:mu} because $\mc{S}(T;t,\sigma)$ is inconsistent. 
				Nevertheless, a relaxed reconciliation map $\mu$ between $(T;t,\sigma)$ 
				and the species tree exists (right). However, this map does not satisfy 
				Lemma \ref{lem:cond-mu}(\ref{item:speci}) 
				since 
				$\mu(a')=A$ and $\mu(\lca_{\Th}(b',c'))$ are comparable.				 
				See text for further details. 
    			}
	\label{fig:counter-bioF-mu}
\end{figure*}

In Section \ref{sec:restMu} we have already discussed that consistency of 
$\mc{S}(T;t,\sigma)$ cannot be used to characterize whether there 
is a reconciliation map  that doesn't need to satisfy (M2.iv) 
for some non-binary gene tree, see Figure \ref{fig:counterBinary}.
In particular, Figure \ref{fig:counterBinary} 
 shows a biologically feasible binary  gene trees (center-left)
for which, however, neither a reconciliation map nor a restricted reconciliation map
exists.  
Therefore, reconciliation maps provide, unsurprisingly, only a sufficient but not necessary condition
to determine whether  gene trees are biologically feasible.
A further simple example is given in Figure \ref{fig:counter-bioF-mu}.
Consider the ``true'' history of the gene tree that evolves along the (tube-like)
species tree in Figure \ref{fig:counter-bioF-mu} (left). The set of extant genes
$\Gen$ comprises $a,a',b,b',c$ and $c'$ and $\sigma$ maps each gene in $\Gen$ to
the species (capitals below the genes) $A,B,C\in \Spe$. For the observable gene
tree $(T;t,\sigma)$ in Figure \ref{fig:counter-bioF-mu} (center) we observe
that  $\mc{R}_0 = \{\rt{ab|c},\rt{b'c'|a'}\}$ and thus, 
one obtains
the inconsistent species triples $\mc{S}(T;t,\sigma) = \{\rt{AB|C},\rt{BC|A}\}$.
Hence, Theorem \ref{thm:charact} implies that there is no species tree for
$(T;t,\sigma)$. Note,  $(T;t,\sigma)$ satisfies also Condition (O3.A). 
Hence, Theorem \ref{thm:charact2} implies that no restricted reconciliation
map to any species tree exists for $(T;t,\sigma)$.
Nevertheless, $(T;t,\sigma)$ is biologically feasible as there is a
true scenario that explains the gene tree.

If Condition (M2.i) would be relaxed, that is, if we allow for speciation
vertices $u$ that $\mu(u) \succeq_S \lca_S(\sT(u))$, then there is a relaxed
$\mu$ from $(T;t,\sigma)$ to the species tree $S$ shown in Figure
\ref{fig:counter-bioF-mu} (right). Hence, consistency of $\mc{S}(T;t,\sigma)$
does not characterize the existence of relaxed reconciliation maps.

\section{Conclusion and Open Problems}
	
	Event-labeled gene trees can be obtained by combining the reconstruction of
	gene phylogenies with methods for orthology and HGT detection. We showed that
	event-labeled gene trees $(T;t,\sigma)$ for which a species tree exists can be
	characterized by a set of species triples $\mc{S}(T;t,\sigma)$ that is easily
	constructed from a subset of triples displayed in $T$.

	We have shown that  biological feasibility of  gene trees cannot be explained 
	in general by reconciliation maps, that is, there are biologically feasible gene trees
	for which no reconciliation map to any species tree exists. 
	Moreover, we showed that consistency of $\mc{S}(T;t,\sigma)$ does not characterize
	the existence of relaxed reconciliation maps. 

	 We close this contribution by stating some open problems that need to be
	solved in future work. 

	\emph{\bf(1)} Are all event-labeled gene trees $(T;t,\sigma)$ biologically feasible?
		
		\emph{\bf(2)} The results established here are based on informative triples provided by the gene trees. 
					  If it is desired to find ``non-restricted'' reconciliation maps (those for which Condition (M2.iv) is not required) 
						for non-binary gene trees
					  the following question needs to be answered:
   					How much information of a non-restricted reconciliation map and a species tree 
		    			is already contained in \emph{non-binary} event-labeled  gene trees $(T;t,\sigma)$?
					The latter might also be generalized by considering  relaxed reconciliation maps (those for which 
					$\mu(x)\succ_S \lca_S(\sT(x))$ for speciation vertices $x$ or any other relaxation is allowed).                

	\emph{\bf(3)} Our results depend on three axioms (O1)-(O3) on the event-labeled
		         gene trees that are motivated by the fact that event-labels can
		         be assigned to internal vertices of gene trees only if there is
		         observable information on the event. The question which
		         event-labeled gene trees are actually observable given an
		         arbitrary, true evolutionary scenario deserves further
		         investigation in future work, since a formal theory of
		         observability is still missing.

	\emph{\bf(4)} The definition of reconciliation maps is 
					by no means consistent in the literature. 
					For the results established here we considered three types of 
					  reconciliation maps, that is, the ``usual'' map as in Def.\ \ref{def:mu}
                 (as used in e.g.\ \cite{THL:11,BAK:12,HHH+12,N+17}), a restricted version (as used in e.g.\ \cite{LDEM:16,Lafond2014})
						and a relaxed version. 
				   However, a unified framework for reconciliation maps is desirable and might
					be linked with a formal theory of   observability.

	\emph{\bf(5)} ``Satisfiable'' event-relations $R_1,\dots,R_k$ are those for which there is a 
						representing gene tree $(T;t,\sigma)$
						such that $(x,y)\in R_i$ if and only if $t(\lca(x,y))=i$. They are equivalent to 
						so-called unp 2-structures \cite{HSW:16}. 
					 	In particular, if event-relations consist of orthologs, paralogs and xenologs only, then
						satisfiable event-relations are equivalent to directed cographs \cite{HSW:16}. 
						Satisfiable event-relations $R_1,\dots,R_k$ are ``S-consistent'' if there is a species tree $S$
						for the representing gene tree $(T;t,\sigma)$ \cite{LDEM:16,Lafond2014}. 
						However, given the unavoidable noise in the input data and possible
						uncertainty about the true relationship between two genes, 
						one might ask to what extent the work of Lafond et al.\ \cite{LDEM:16,Lafond2014}
						can be generalized to determine whether given ``partial'' event-relations 
						are S-consistent or not. 
						It is assumable that subsets of the informative
						species triples $\mc{S}(T;t,\sigma)$ that might be directly computed from such
						event-relations can offer an avenue to the latter problem.
						Characterization and complexity results for 
						``partial'' event-relations to be satisfiable have been addressed in  \cite{HW:16a}.

		\emph{\bf(6)} In order to determine whether there is a time-consistent reconciliation map
					for some given event-labeled gene tree and species 
					trees fast algorithms have been developed \cite{N+17}. 
					 However, these algorithms require as input a gene tree $(T;t,\sigma) $
					\emph{and} a species tree $S$. 	
					A necessary condition to a have time-consistent (restricted) reconciliation map
					to some species tree is given by the consistency of the species triples $\mathcal{S}(T;t,\sigma)$. 
					However, in general there might be
					exponentially many species trees that display $\mathcal{S}(T;t,\sigma)$ for which
					some of them may have a time-consistent reconciliation map with $(T;t,\sigma)$ and some might have
					not (see Figure \ref{fig:least} or \cite{N+17}). Therefore, additional constraints to determine whether 
					there is at least one species tree $S$ with time-consistent
					map, and if so, construct $S$, must be established. 

		\emph{\bf(7)}	A further key problem is the identification of horizontal transfer events. 
						In principle, likely genes that have been introduced into a genome by HGT
					   can be identified directly from sequence data \cite{soucy2015horizontal}. Sequence
						composition often identifies a gene as a recent addition to a genome. In the absence
					   of horizontal transfer, the similarities of pairs of true orthologs in the species pairs
						(A,B) and (A,C) are expected to be linearly correlated. Outliers are likely candidates for
						HGT events and thus can be ``relabeled''. However, a 
						more detailed analysis of the relational properties of horizontally 
						transferred genes is needed.

\bibliographystyle{plain}
\bibliography{biblio}
\end{document}